\documentclass[conference]{IEEEtran}
\IEEEoverridecommandlockouts
\usepackage{graphicx}
\usepackage{array, tabularx}
\usepackage{amsmath,amsthm,amssymb}
\usepackage[ruled,vlined]{algorithm2e}
\usepackage{cite}
\usepackage[hidelinks]{hyperref}
\usepackage[shortlabels]{enumitem}
\usepackage{url}
\usepackage{microtype}
\usepackage{hyperref}
\usepackage{caption}
\usepackage{xspace}
\usepackage{float}

\newtheorem{lemma}{Lemma}[section]
\newtheorem{example}{Example}[section]
\def\BibTeX{{\rm B\kern-.05em{\sc i\kern-.025em b}\kern-.08em
    T\kern-.1667em\lower.7ex\hbox{E}\kern-.125emX}}

\begin{document}

\title{A Multiple Snapshot Attack on Deniable Storage Systems\\
\thanks{This research was supported in part by the National Science Foundation grant number IIP-1266400, award CNS-1814347, and by the industrial partners of the Center for Research in Storage Systems.}
}
\author{\IEEEauthorblockN{Kyle Fredrickson, Austen Barker, and Darrell D. E. Long}
\IEEEauthorblockA{\textit{University of California, Santa Cruz} \\
}
}

\newcommand{\etal}{\emph{et~al.}\xspace}
\newcommand{\afs}{Artifice\xspace}
\newcommand{\eg}{\emph{e.g.}\xspace}

\maketitle
\thispagestyle{plain}
\pagestyle{plain}

\begin{abstract}
While disk encryption is suitable for use in most situations where confidentiality of disks is required,
stronger guarantees are required in situations where adversaries may employ coercive tactics to gain access to
cryptographic keys.
Deniable volumes are one such solution in which the security goal is to prevent an adversary from discovering that there is an encrypted volume.
Multiple snapshot attacks, where an adversary is able to gain access to two or more images of a disk, have often been proposed in the deniable storage system literature;
however, there have been no concrete attacks proposed or carried out.
We present the first multiple snapshot attack, and we find that it is applicable to most, if not all, implemented deniable storage systems.
Our attack leverages the pattern of consecutive block changes an adversary would have access to with two snapshots,
and demonstrate that with high probability it detects moderately sized and large hidden volumes, while maintaining a low false positive rate.
\end{abstract}


\section{Introduction}

Conventional disk encryption has been very successful at ensuring the confidentiality and integrity of disks.
In many applications reducing these properties to corresponding key management problems is sufficient to ensure security.
However, when faced with adversaries capable of using coercive tactics to reveal keys, conventional encryption comes up short.
In these situations the possession of an encrypted drive can be considered suspicious enough to justify coercive attacks
euphemistically known as rubber hose attacks.
These situations require that the mere presence of encrypted data be unknown to the adversary to avoid the use of coercive tactics.
Deniable storage systems have been developed in response to this threat.
These systems are used to create a hidden volume on the user's device,
the existence of which is plausibly deniable.
One common strategy among deniable storage systems is to encrypt data and randomly write it throughout the disk's free space.
Assuming that the free space is filled with pseudorandom bytes, this renders hidden data indistinguishable from unallocated blocks.
This achieves the goal of deniability when an adversary is restricted to viewing the disk at one point in time and there are no other
sources of information leakage.

A multiple snapshot attack is an attack on a hidden volume where an adversary is able to gain access to a machine and make observations at two or more points in time.
Information gained from comparing these snapshots would then be analyzed for abnormalities that could imply the existence of a hidden volume.
In many circumstances where a deniable volume may be used multiple snapshot attacks are feasible.
As an example, suppose a journalist is entering a repressive country with the intent to exfiltrate some data.
The adversary first takes a snapshot of all devices entering the country,
the journalist collects the data to be exfiltrated and constructs a hidden volume to hide it,
then the adversary takes another snapshot when the device leaves the country.
This leaves the pattern of changes on the file system as a new source of information leakage.

While multiple snapshot attacks on hidden volumes have been described in the literature and defenses against them proposed~\cite{Zhou2004, Blass2014, Chakraborti2017, Chen2019},
to our knowledge there has never been a thoroughly described or attempted example of this class of attack.
Among the reasons for this are not only the difficulty in obtaining enough disk images to establish what a \emph{normal} pattern of changes is,
but also the difficulty in identifying meaningful features in observed change patterns that would be invariant across normal use cases.

Using features derived from the change patterns on disks, we have defined and implemented the first multiple snapshot attack against deniable storage systems.
The main contribution of our work is to affirm the relevance of multiple snapshot attacks to deniable storage systems by demonstrating a practical attack,
while also identifying the limitations of our techniques.
We propose analyzing the distribution on the lengths of consecutive block changes, which we call {\it chains}, as a new metric for quantifying disk behavior,
and leverage this information to distinguish between disks containing hidden volumes and those that do not.
Specifically we attack \afs~\cite{barker-msst20}, a deniable storage system,
but our attack if broadly applicable, and is able to reveal most, if not all, implemented deniable storage systems.
In response to this attack we propose additional design requirements for deniable systems,
and discuss the implications of these requirements on \afs.
Through this work we hope to guide the design and implementation of future systems,
improving their security even against powerful adversaries with the ability to gain access to devices at several points in time.

We first give background on deniable storage systems, going into detail on \afs, and describe past attacks on these systems.
We then propose our attack, describe our data collection and simulation methods, and give results of the attack on our dataset.
Finally, we describe mitigation schemes and conclude.
\section{Background and Related Work}

As previously stated, the goal of a deniable storage system is to conceal the existence of a
volume from an adversary that would otherwise view an encrypted volume as suspicious and consequently
coerce the user to disclose their secret keys.
With a hidden volume the user is able to
reveal keys to known encrypted drives while allowing the user to deny that hidden volumes
exist on the device.

Many deniable or steganographic storage systems have been proposed in past years. One of the key concerns of
these systems has been defending against multiple snapshot attacks. In this section we will
describe existing deniable storage systems and attacks against those systems.

\subsection{Deniable and Steganographic Storage Systems}
Anderson, \etal~\cite{Anderson1998} were the first to propose a steganographic file system and described
two possible approaches. The first approach consists of a set of cover files filled with random information, of which a subset are combined
with hidden files using an additive secret sharing scheme. The second construction hides data within the unallocated space of another file system.
Although the proposal lacked an implementation of the two ideas,
most deniable storage systems follow the second approach.

McDonald and Kuhn implemented Anderson \etal's second scheme as a Linux file system based on
\texttt{ext2} known as StegFS~\cite{McDonald1999}. Although the scheme provides deniability when
the adversary can only view the device once, McDonald \etal's StegFS and similar systems like TrueCrypt~\cite{Truecrypt} or Mobiflage~\cite{Skillen2013}
cannot provide a reasonable defense against an adversary that can view the disk multiple times and compare snapshots.

Pang, \etal~\cite{Pang2003, Zhou2004} implemented a variant of StegFS that attempts to defend against multiple snapshot adversaries
by performing dummy operations to the disk that obscure hidden writes. Similar approaches to this have been implemented
using derivatives of Oblivious RAM~\cite{Goldreich1987, Goldreich1996} or similar techniques to render accesses to a hidden
volume indistinguishable from accesses to the public volume~\cite{Blass2014, Chakraborti2017, Chen2019}. Although these
approaches render hidden and public writes indistinguishable, they incur other abnormalities such as fully random write patterns
not exhibited by common file systems and significant performance losses that would betray the existence of a hidden volume or the ability to construct a hidden volume.

A perennial weakness of many of these systems is that they do not take measures to hide the existence of the software used to access the hidden volume,
often claiming that widespread adoption would ensure the ability to create a hidden volume is not suspicious.
The fact that disk encryption, though widespread, can still be seen as suspicious by many organizations calls this assumption into doubt.

\subsection{Attacks on Steganographic Storage}

Most existing storage systems assume that the adversary has the ability to closely analyze characteristics of the user's device for any
clues that the device contains a hidden volume.
With respect to multiple snapshot analysis and related attacks,
previous work~\cite{Blass2014, Zhou2004} has categorized adversary capabilities into three general categories.

\begin{itemize}
	\item \textbf{Single snapshot} A single snapshot adversary can only view a device once. All existing deniable systems have
		some level of resilience against this sort of attack.
	\item \textbf{Multiple snapshot} In this case the adversary can view the device two or more times. These snapshots
		can be compared, and the changes analyzed for anomalies that could reveal a hidden volume on the user's device.
		It is important to note that in this class of attack the adversary is only able to view static snapshots of the disk
		a limited number of times.
	\item \textbf{Continuous observation} This is when an adversary has the ability to continuously observe writes to the user's device
		or make a snapshot of the device for each write. This adversary capability is sometimes also called \emph{continuous traffic analysis}.
		This sort of attack would likely require a form of malware to be installed on the user's device for the purposes of information
		gathering. OS level spyware would be able to monitor writes to the user's device.
\end{itemize}

Proposed attacks and serious attempts to break deniable storage systems are limited when compared to the variety of
proposed systems. One vulnerability found by Czeskis \etal when analyzing TrueCrypt was its susceptibility to
leaking information from the hidden to public volumes through operating system utilities and common applications such as
word processors. Another work by Troncoso \etal describes and implements a continuous traffic analysis attack against
Pang's StegFS that exploits repeated write patterns that correspond to writes made to a hidden volume. While this
attack shows the effectiveness of monitoring disk operations in finding a hidden volume, it is assumed that the
adversary has enough power to continuously observe the disk.
The multiple snapshot adversary is considerably weaker than the continuous observation adversary,
yet there has been no previous work towards demonstrating a multiple snapshot attack against a deniable storage system.





\subsection{Artifice}

For the purposes of testing a multiple snapshot attack we will use \afs as our deniable storage system.
\afs follows the common model of hiding information in unallocated blocks but
with a few additional features that address some pitfalls of other existing file systems~\cite{barker-msst20}. Most importantly
it addresses problems with hiding its driver software and provides a layer of protection against malware and
information leakage by putting the \afs driver on a separate Linux live USB drive. To access a hidden volume,
the user boots into an \afs-aware OS contained on this drive instead of the normal public OS. This isolation does
not leave behind suspicious drivers on the user's machine and mitigates the impact of malware and information leakage.

\afs writes data by splitting data blocks into pseudo-random \emph{carrier blocks} using an information
dispersal algorithm (IDA) such as Shamir Secret Sharing~\cite{Shamir1979}. These carrier blocks are then uniformly
distributed throughout the unallocated space of the drive, which is assumed to be full of pseudorandom blocks due
to a secure deletion utility or similarly deniable means.
Since the public file system is not aware of \afs, it is imperative to protect the carrier blocks from accidental overwrite.
IDAs provide \afs overwrite tolerance by writing redundant carrier
blocks in excess of the number needed to normally reconstruct the data.
This allows \afs to carry out self-repair operations whenever accessed by the user
and increases the probability that an \afs instance will survive many writes made by the public file system.

Currently, \afs aims to address the problem of multiple snapshot attacks through writing hidden blocks
under the guise of a suitable deniable operation, such as defragmentation, where the contents of a disk are relocated to be contiguous,
routine file deletion, or by operational security measures that render previous snapshots useless,
such as reinstalling the public operating system or wiping the storage prior to constructing an \afs instance.

\section{Attack Framework} \label{section:attack}

\begin{figure}[t]
    \centering
    \def\svgwidth{\columnwidth}
    \scalebox{1.0}{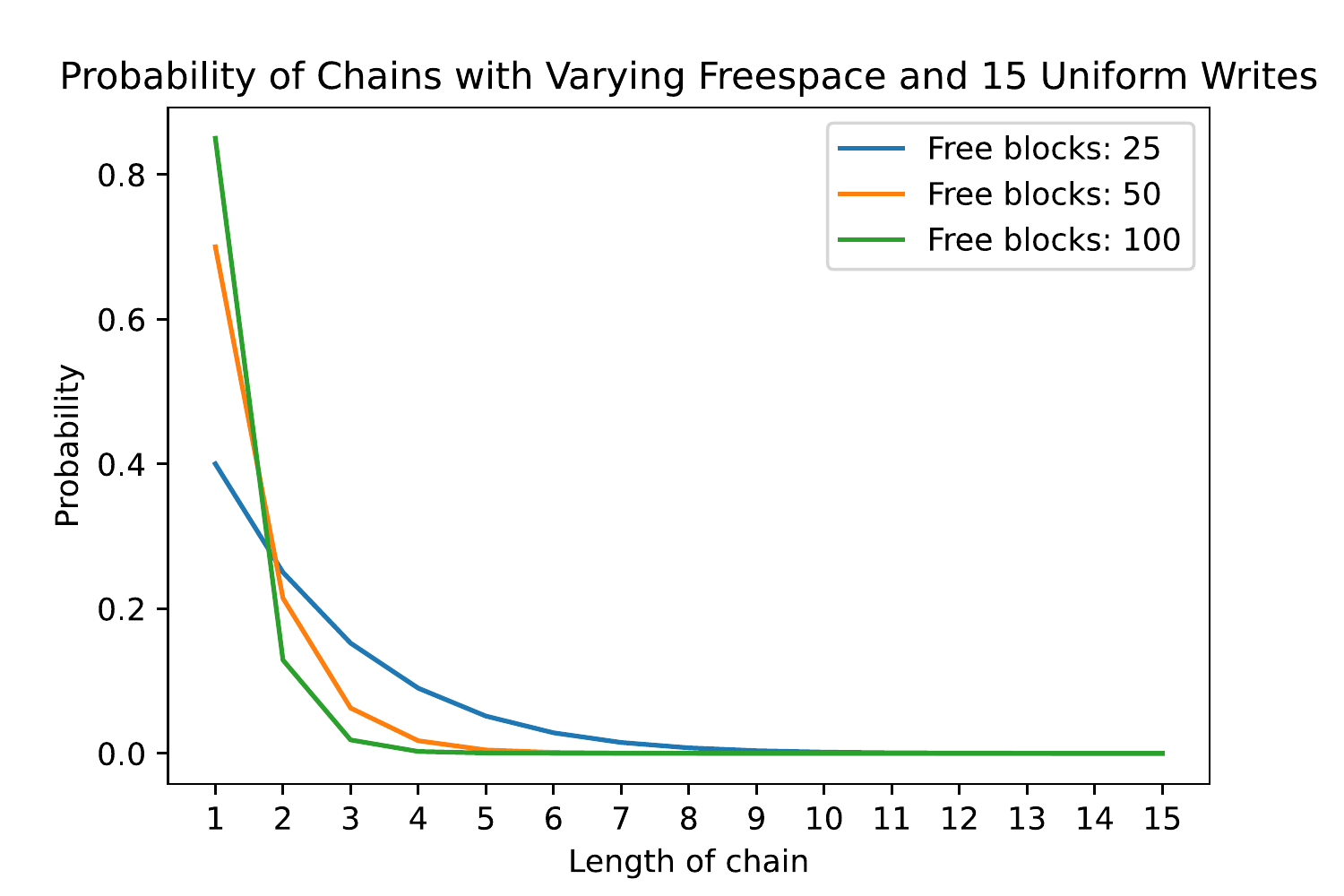}
    \caption{Theoretical probability of consecutive changes when changes are made uniformly. The probability of $c$ consecutive changes degrades very quickly, especially when the free space is large.}
    \label{fig:prob-cons-changes}
\end{figure}

\begin{figure}[t]
    \centering
    \def\svgwidth{\columnwidth}
    \scalebox{1.0}{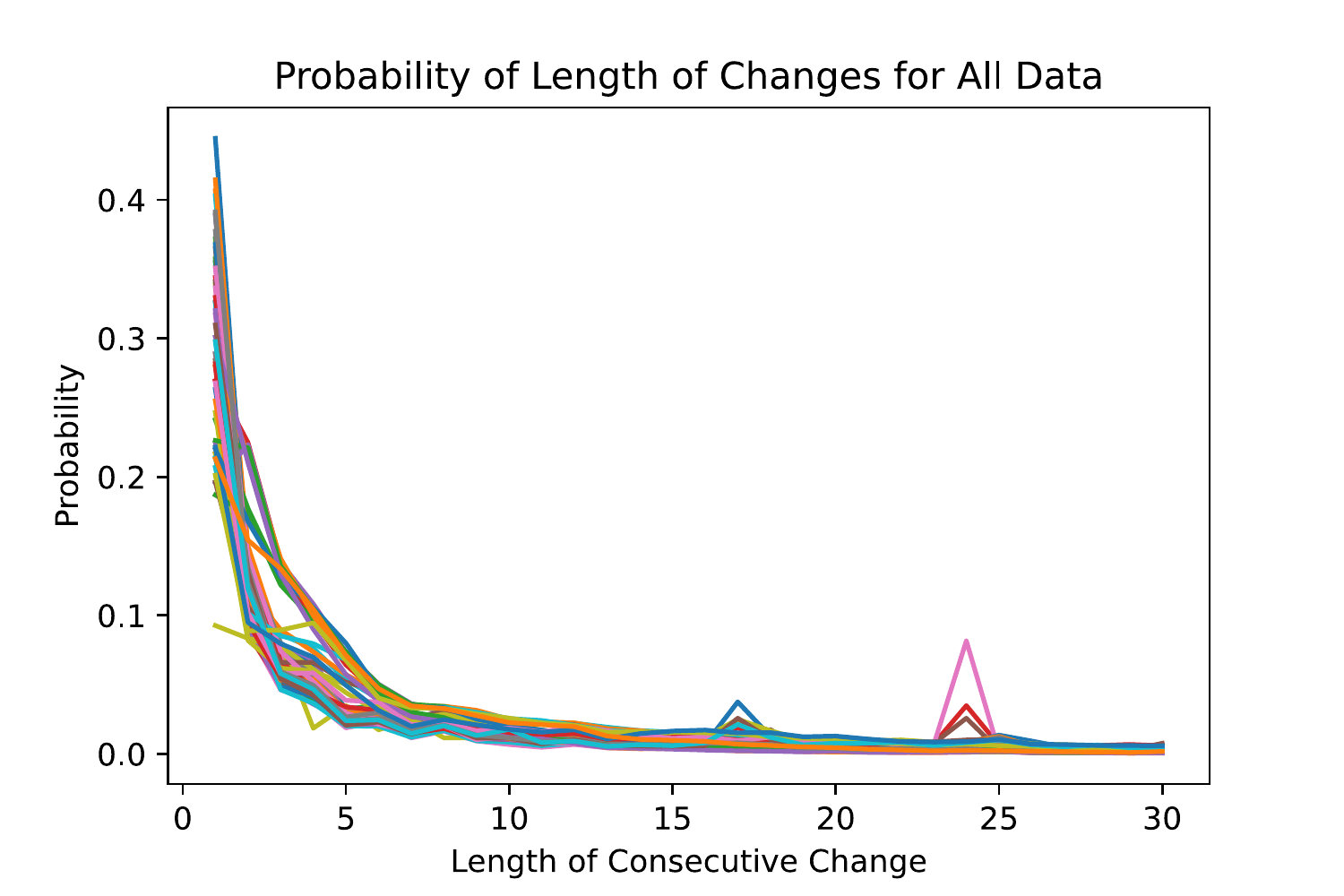}
    \caption{Empirical probability of consecutive changes for 52 change records where the changes are made by an \texttt{ext4} file system.}
    \label{fig:prob-over-dataset}
\end{figure}

The primary goal of an attack on a deniable storage system is not to recover the plaintext files, but to discover the presence of hidden volumes.
In \afs's adversary model, it is assumed that the adversary can coerce a user to reveal their keys provided there is sufficient suspicion that a user
possesses a hidden volume~\cite{barker-msst20}. The success of a particular attack then is dependent on how well the attack can discriminate
between disks containing hidden volumes and those that do not.
In particular, the rates of false positives and false negatives should both be very low for an attack to be considered successful.

As noted most deniable storage systems achieve their aims by encrypting and splitting data into redundant shares and writing these blocks uniformly
on the disk~\cite{barker-msst20, Pang2003}.
This increases the odds of survival of the files when the public file system, which is unaware of their existence, makes its own writes.
Other approaches write public and hidden data psuedorandomly so that public and hidden writes are indistinguishable~\cite{Blass2014, Chakraborti2017}.

The commonality between these approaches is that deniable storage systems make many uniform writes.
Depending on the size of the free space, writes made uniformly are very likely to result in isolated changes on disk,
which we call {\emph singletons}.
This is in contrast to normal file systems, which do not make their writes uniformly,
and are much more likely to make writes that are part of longer strings of consecutive changes that we call {\emph chains}.
We call a chain of $c$ consecutive changes a $c$-chain.

\begin{example}
Say that in a change record a 1 denotes a change and a 0 denotes no change between two snapshots.
Then in (\ref{ex:change-record}) there are two singletons (or 1-chains) and one 3-chain.

\begin{equation}\label{ex:change-record}
    \begin{tabular}{ |c|c|c|c|c|c|c| }
        \hline
        1 & 0 & 1 & 0 & 1 & 1 & 1\\
        \hline
       \end{tabular}
\end{equation}
\end{example}

An adversary observing the change records produced from a pair of disk snapshots would be able to observe the lengths of the chains those changes produce.
Crucially if the disk contains a hidden volume, the adversary would also see changes made by writes to that volume.
In Appendix \ref{appendix} we give the theoretical distribution of chains given the size of the disk and the number of changes made.
We find they are distributed according to lemma \ref{change-lemma}.
Fig. \ref{fig:prob-cons-changes} shows that as the free space grows relative to the number of writes, the probability of a singleton increases.
Fig. \ref{fig:prob-over-dataset} shows that for real disks, the probability of a singleton is much smaller, and the tail of the distribution is typically much heavier.
Together they show the disparity between the distributions of chains due to a hidden volume and the chains due to a public file system.
This becomes more pronounced as the hidden volume makes more writes.
Our task then is to construct features to distinguish between the distribution of consecutive changes made by a public file system and changes made by a public file system {\it and} a hidden volume.

To carry out this attack, we assume that the attacker has access to a large set of disk images, both from disks that contain hidden volumes and from those that do not.
Images in this set will be organized into pairs of images from the same disk at different points in time. Comparing these pairs of images results in a change record for a given disk.
We will call pairs that do not have an instance of a hidden volume clean and those that do dirty.
Assuming our adversary is well-funded and motivated these data requirements are easily attainable.

Since for our proposed attack we only need to determine whether individual blocks have changed,
we can take snapshots of the clean and dirty disks in a space efficient manner by hashing each block on the disk, and constructing a Merkle tree \cite{merkle1987digital} over the hashed blocks.
This gives us a very efficient method of finding changes, and producing change records.
These change records are further processed into lists of integers, $D_i$, recording the lengths of chains found in each change record.
We will denote these lists of chains $\{D_i \} = \mathcal{D}$.

From $\mathcal{D}$ the adversary has several options for constructing an arbitrary $n$ number of features for use in a classification algorithm.
The first is to remove a set of clean disks that we will call $\mathcal{C}$.
Using $\mathcal{C}$ the adversary estimates the probability of $c$-changes from $1$ to $n$, the number of features, for each $D_i$ by counting the occurrences of each $c$-chain and dividing by the total number of chains in $D_i$.
Using these probabilities on $\mathcal{D} - \mathcal{C}$, the adversary can then estimate the probability of a disk containing more than $k$ $c$-chains with the cumulative distribution function (CDF) of the binomial distribution,
$F(k; n, p_c)$ where $n$ is the length of the change record and $p_c$ is the estimated probability of a chain of $c$ consecutive changes.
We will denote the event of an adversary observing $k$ consecutive changes of length $c$ as $X_c$.
Then,
\begin{equation} \label{binom-cdf}
P(X_c > k) = 1 - F(k; n, p_c).
\end{equation}
Using these values we construct our final features $\mathcal{F}$ for $\mathcal{D}$.
This method has greater sensitivity to small variations in probabilities for small disks;
however, on large disks it tends to underflow when computing equation (\ref{binom-cdf}).
We give pseudocode for this in algorithm (\ref{feature-construction}).

For large datasets, we take a simpler approach to computing $\mathcal{F}$ by
computing the probabilities of chains of length $1$ to $n$ for each $D_i$ and feeding these to our classifier.

Recall that the dataset is entirely constructed by the adversary,
so it has ground truth labels describing whether each row in $\mathcal{F}$ corresponds to a disk containing a hidden volume or not.
The adversary now trains a supervised classification algorithm on $\mathcal{F}$ split into standard train and test sets.
On new pairs of disks the adversary runs through the feature construction process, then runs the classification algorithm on those feature and responds accordingly.
Furthermore, if the adversary can confirm that some disks did in fact contain a hidden volume, it can update its model using various online learning techniques \cite{shalev2011online}, further refining its model.

As a note on selection of the number of features,
we observe from lemma \ref{change-lemma} that when the free space on a disk is large relative to the number of writes, the vast majority of writes will result in singleton changes.
As a consequence, the adversary could learn based on a single feature derived from singletons.
This may be desirable in some situations for the sake of efficiency;
however, \afs could simply be modified so that when writing blocks they would be grouped in chains of two or more, thereby defeating the attack as described.
The adversary can in turn thwart this countermeasure by increasing the number of features, $n$. We go into more detail regarding this problem in Section \ref{mitigation}.

\begin{algorithm}
\SetAlgoLined
\KwIn{$\mathcal{D}$, a set of processed disks; $\{p_1, ..., p_c\}$, the estimated probabilities of consecutive changes of length $1...c$.}
\KwOut{$\mathcal{F}$, a $|\mathcal{D}| \times c$ matrix.}
    $\mathcal{F}:=\{\}$\\
    \ForEach{$D_i \in \mathcal{D}$}{
        $f := \{\}$\\
        \ForEach{$c_i \in \{1...c\}$}{
            $k := $ number of $c$-consecutive changes in $D_i$\\
            $a := 1 - F(k;|D_i|,p_{c_i})$\\
            $\mathsf{append}(a, f)$
        }
        $\mathsf{append}(f, \mathcal{F})$
    }
    $\mathsf{return}$ $\mathcal{F}$
 \caption{Feature Construction. This method is suitable for smaller disks, where greater sensitivity is required, but suffers from underflow on large disks.}\label{feature-construction}
\end{algorithm}

\section{Data Collection and Experiment Methodology} \label{data-collection}

As noted previously, one of the challenges of carrying out a multiple snapshot attack is the availability of pairs of disk images.
These are necessary to learn what normal chains look like. While collecting hundreds, or thousands, of disk images may be feasible for a nation state level adversary (or a large IT department),
we were unable to collect such a large amount of data.

Instead, we have collected several months worth of snapshots from a 1\,TB NVMe SSD formatted with \texttt{ext4} and in use as the boot disk of a desktop computer running Ubuntu 18.04.
We collected 53 snapshots in total, giving us 52 change records\footnote{This data is available at \href{https://files.ssrc.us/data/disk-change-data.zip}{https://files.ssrc.us/data/disk-change-data.zip}. The code for these experiments is available at \href{https://github.com/ucsc-ssl/multiple-snapshot-attack}{https://github.com/ucsc-ssl/multiple-snapshot-attack}.} which we have made publicly available.
By observing the distribution of lengths of changes over our collected data (Figure~\ref{fig:prob-over-dataset}),
and the theoretical distribution of consecutive changes (Figure~\ref{fig:prob-cons-changes}) when changes are made uniformly,
the potential strangeness of a disk running a deniable volume becomes clear.

Because just 52 data points are insufficient to train a classifier, and moreover, would be inconclusive regarding performance of that classifier,
we instead used this data to generate a synthetic dataset on which to train and test our classifier.
While in the real world different file systems may produce different patterns of changes,
this does not change the reality that a deniable volume writing many single blocks,
and thereby causing many singleton chains in the change record,
would be considered abnormal regardless of the public file system in use.
Though the fact that our attack only utilizes data from \texttt{ext4} file systems is a limitation,
no file system in widespread use writes blocks randomly, so we expect our attack will generalize to other data sources.

Our experiments are conservative in terms of the operational security measures the user of a hidden volume might take.
We discuss them here to motivate our experimental design.

There are several things that the user of a deniable volume could do to decrease the odds of detection.
To start, assume that a single snapshot has been taken, and no deniable volume yet exists on the drive.
A prudent user would make many changes to the disk through the public file system.
The reason for this being that if there are overwhelmingly many chains distributed according to normal disk behavior,
the singletons made by writing to the hidden volume could be made to look like noise.
To illustrate this, consider a user that does not produce a single change through the public file system after the deniable volume is created and written to.
In this case the adversary would see, after taking a second snapshot and computing the differences, only chains produced by the deniable volume.
These chains would be principally singletons, which would surely be conspicuous.
As an extreme measure the user could simply wipe the disk and then create the deniable volume, but this may be considered suspicious or may be undesirable for other reasons.
As a less drastic alternative, the user could produce an overwhelming number of changes to the disk through the public file system.
In our experiments we chose this middle ground by using our real data to simulate 25\,GB worth of public changes
on a 1\,TB disk with 100\,GB of free space.
This produces a sufficient number of public changes to hide private changes while still behaving as a normal user might.

Each pair of disk snapshots can be regarded as producing a distribution over consecutive change lengths,
so in order to construct our synthetic dataset we simply draw chains from these distributions.
Realistically, the fraction of disks containing hidden volumes would be relatively small, and the size of hidden volumes would also be variable.
For our hidden volumes, we assume that we have instances that are from 250\,MB to 1.25\,GB in increments of 250\,MB.
In addition to its realism, this allows us to determine a point at which the number of uniform writes becomes conspicuous.
For our \afs parameters we chose those that minimized the number of writes, while achieving survival probabilities over 80\% with 25\,GB of cover changes.
This led us to copy data blocks 6 times, where the survival of a single block is sufficient for reconstructing the data.
Since our adversary is able to generate an arbitrary number of disk snapshots with and without hidden volumes,
in order to allow our classifier to learn to distinguish disks more quickly our training set contains an even split of disks with and without hidden volumes.
However, to reflect the rarity of hidden volumes in the real world only 5\% of our test set contains disks with a hidden volume.
We generate a training set of size 10000 and a test set of size 2500.
We repeat this generation, training and testing cycle 100 times to ensure the reliability of our results.

\section{Results}\label{results}
By implementing the experimental methodology described in Section~\ref{data-collection} and
running it against our dataset we collected a set of results
that show the efficacy of our proposed multiple snapshot attack.
It should be noted from the start that at the core of our implementation is a
simple logistic regression that takes only the probability of a single block change on the disk into account.
In choosing to implement such a simple learning algorithm we highlight the distinguishing power of analyzing consecutive block changes in
detecting anomalous disk behavior.

We collected five different metrics on our classifier: accuracy, precision, recall, false positive rate, and false negative rate.
Because only 5\% of our test set contains \afs instances accuracy is not a very informative metric, and we include it for completeness only.
Precision is the ratio of true positives to predicted positives.
Recall is the ratio of true positives to true positives and false negatives,
giving the ratio of \afs instances that were identified from the test set.
False positive rates and false negative rates are useful for understanding how frequently our classifier makes errors in both directions.
We consider false classification rates to be the most important metrics for an attacker.

In our experiments the 250\,MB \afs instances were often able to pass undetected, implying that 25\,GB of cover changes were sufficient to hide these volumes.
However, the largest four sizes were reliably detected,
with the three largest sizes, 0.75\,GB, 1.0\,GB and 1.25\,GB, being detected nearly 100\% of the time.
This highlights a feature of our attack, namely that for 25\,GB of cover changes every \afs instance above a certain size will be detected with high probability.
This is because of our use of logistic regression, and because the probability of singletons is so overwhelming.
Eventually as free space fills up, \afs will begin to make changes that are parts of longer chains, but if the reconstruction threshold is low, this will severely impact the survivability of the volume.

\begin{figure}[t]
    \centering
    \def\svgwidth{\columnwidth}
    \scalebox{1.0}{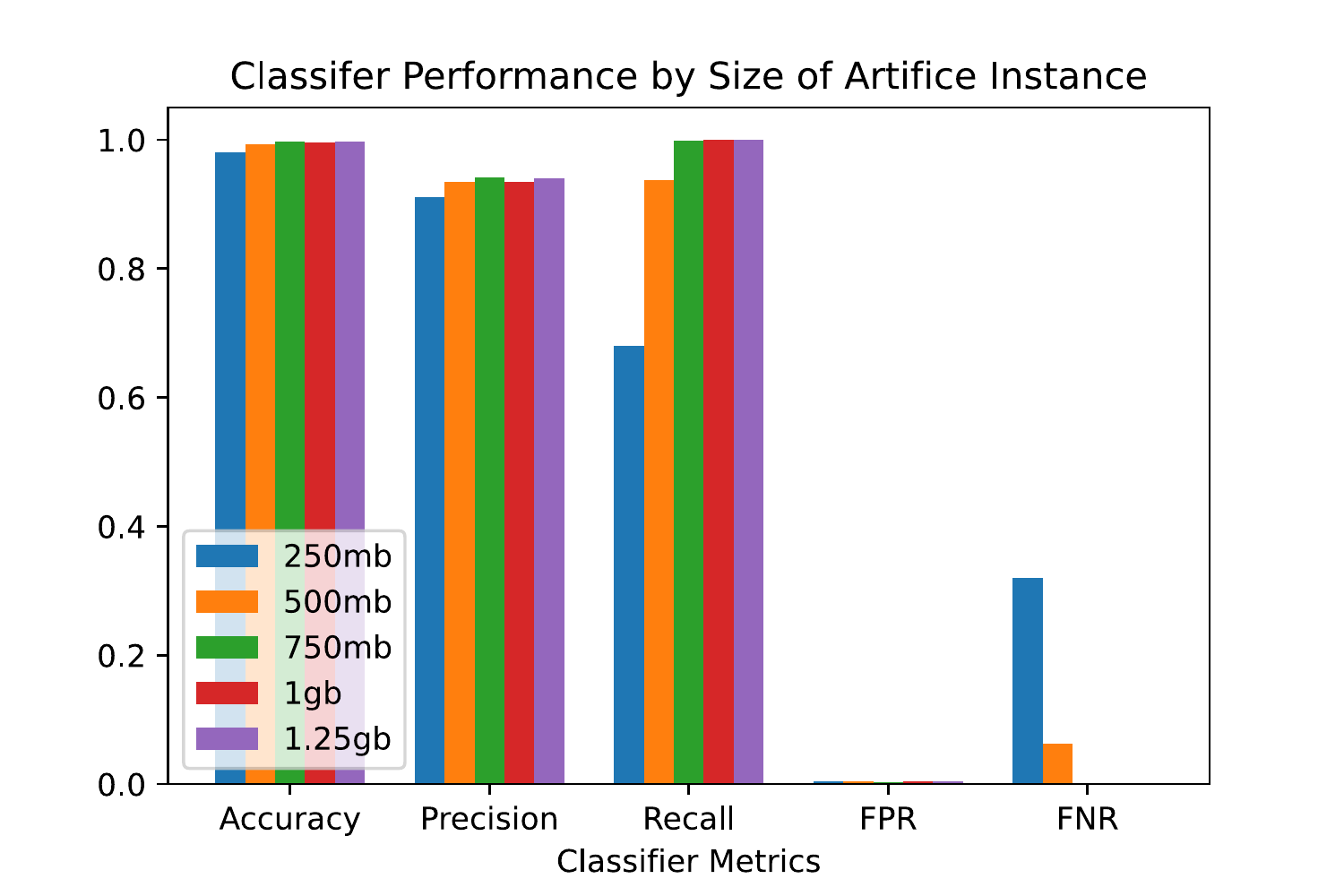}
	\caption{Hidden volumes in excess of 0.75\,GB are always identified successfully.}
    \label{fig:results}
\end{figure}

The presence of a point where the FNR becomes negligible also offers an explanation for the relatively constant false positive rate.
This being that a certain percentage of simulated clean disks will naturally have disproportionately many singletons and thus get misclassified as containing an \afs volume.
Interestingly, there were very few disks that naturally had enough singletons to exceed the learned threshold.

Future work may combine the features we use in this attack with other features. Such as what proportion of
changes are made to blocks in free space versus allocated blocks. The introduction of more features would
serve to better characterize disk behavior and further improve the efficacy of the attack.

\begin{figure}[t]
\begin{center}
    \begin{tabular}{ c c c c c c }
        Size (GB) & Acc. & Precision & Recall & FPR & FNR \\
        \hline
        0.25 & 0.981 & 0.911 & 0.680 & 0.004 & 0.320 \\
        0.5 & 0.993 & 0.934 & 0.937 & 0.004 & 0.063\\
        0.75 & 0.997 & 0.942 & 0.999 & 0.003 & 0.001 \\
        1.0 & 0.996 & 0.935 & 1.0 & 0.004 & 0.0 \\
        1.25 & 0.997 & 0.939 & 1.0  & 0.004 & 0.0
    \end{tabular}
\end{center}
\caption{Numeric results from our experiments. These results are averaged over 100 runs. Furthermore, 95\% confidence intervals for all figures vary only in the thousandths.}
\label{fig:results-table}
\end{figure}

\section{Attack Mitigation} \label{mitigation}


Artifice's authors propose an operational security based approach to defend against multiple snapshot attacks~\cite{barker-msst20}. The rationale being that if the user can produce a
deniable reason for changes to the entire disk, such as re-installing the operating system or defragmentation, then the adversary's previously gathered data would
be rendered useless. This defense relies heavily on the ability of the user to out-maneuver the adversary and will not always be practical. As a result, it would be prudent to develop
some other countermeasures against our proposed attack and future types of snapshot analysis.

We can observe in Figure \ref{fig:results} that the FNR of our classifier decreases as the effective size of the \afs instance increases while the amount of data written
by the public volume remains the same. With this in mind we can conclude that the operational strategy proposed in Section \ref{data-collection}, to keep the proportion of hidden to public writes
sufficiently skewed in favor of the latter, is a viable strategy.
Unfortunately, this severely limits the effective size of \afs.

While our proposed mitigation technique is a promising means of defense against snapshot analysis it is still worthwhile to explore the possibility of a mechanism that does not
require such close user involvement. In order for a deniable storage system to defend against a snapshot attack without operational measures, the designer of such a system could attempt to mimic
the distribution of writes that the public file system is making. Similar to our operational approach, the survivability of a hidden volume may be impacted by changing the distribution
of writes from uniform to something more closely resembling a normal file system.
In the case of \afs, the reason that it writes blocks uniformly is to ensure that any writes made to the public file system are unlikely to destroy
too many \afs blocks. As such, system designers need to be careful when implementing countermeasures to maintain a high probability of survival when using the
public file system while still successfully mimicking its write behavior.

In our proposed attack, we use only one feature associated with single block changes to the disk.
A naive approach would be to simply write all blocks in pairs, producing consecutive changes of two blocks.
This would defeat our attack, but the adversary could expand its feature set to include changes two blocks long and so on, forcing the user of the hidden volume to mimic the public file system or risk detection.
Furthermore, given the myriad measurable features in the average file system and that the adversary is unlikely to publish details of its attack,
we can conclude it would be difficult to accurately determine what features are and are not being tracked by the adversary, making mimicry the safest option.

In order to accurately mimic expected access patterns, much more work needs to be done to quantify exactly what a change pattern for a disk without a deniable volume looks like.
A significant body of work has been published in the context of network steganography and pattern mimicking cryptography and could be used
to inform designers of future multiple snapshot resistance systems.
For instance, initial work into format transforming encryption by Dyer\etal~\cite{dyer2013protocol} showed that it was possible to efficiently encrypt data so that it would conform
to a target regular expression. An application of this is found in censorship resistant networking. In this case a user might be running a blacklisted protocol, such as Tor,
but transforms the protocol messages in such a way that they look like HTTPS. Similar techniques could be applied to deniable
storage to disguise suspicious write patterns. Unfortunately, there is also evidence that more capable adversaries utilizing more sophisticated
attacks can easily identify these false protocols. Houmansadr \etal~\cite{Houmansadr2013} argue that it is unlikely that one could mimic a protocol perfectly without running
the actual protocol, because there are often sub-protocols one would also need to mimic or differences in implementations that allow for version fingerprinting.
Assuming snapshot analysis becomes more sophisticated it is likely that mimicry techniques applied to deniable storage would also need to evolve.

Since \afs relies on pseudorandom data in free space and the use of secure deletion utilities to produce cover changes,
one way to potentially sidestep the issues of artificial mimicry could be found in actually using these deleted files.
\afs could keep track of changes in free space on the public file system, and when it sees a block added to free space it could overwrite this with an \afs block.
As noted above this is not uniformly written and so may risk corruption of files, but it also may provide stronger mimicry guarantees than other methods.
We leave it to future work to investigate this and the other techniques we have presented for mitigation.

\section{Conclusion}
We have demonstrated the first implemented multiple snapshot attack against deniable storage systems.
In doing so we showed the broad usefulness of change records to an adversary seeking to detect hidden volumes.
Furthermore, we presented a concrete way to analyze changing disks, even when the contents and sizes of those disks are radically different.
Through our data collection we were able to show that this measure of computing consecutive change lengths of a disk is relatively stable.
For our data we collected over fifty images of \texttt{ext4} disks and used records of changes across these images to train a classifier on probabilities of chains.
We demonstrated that this classifier is able to differentiate between disks containing a hidden volume from disks without a hidden volume in a variety of configurations. In the process
we have also identified limitations to our technique and from those limitations have proposed possible countermeasures to our attack.

As future work we would like to gather substantially more data from more varied sources.
In addition to covering the major file systems, collecting data from many different types of computer users would give us greater confidence in the stability of our metrics.
Finally, the attack we propose has great potential for expansion.
Many more features could be included, such as number of writes and location of writes in free and used space.
Adding additional features would be trivial, and have the potential to further improve the performance of the attack,
lowering the false positive rate, and the size of hidden volumes that can be expected to evade detection.


\Urlmuskip=0mu plus 1mu\relax
\bibliographystyle{IEEEtran}
\bibliography{bibliography}


\appendix

\subsection{Theoretical Probability of Consecutive Changes} \label{appendix}

Suppose we have an array, $A$, of size $n$ and we make $k$ changes to it at random, where $1$ denotes a change and $0$ denotes no change.
What is the probability of selecting a chain of exactly $c$ consecutive changes?
We give an example and then give the general statement and proof.

\begin{example}\label{a:example}
    When $n=7$ and $k=4$, what is the probability of drawing a chain of length 2 from the array?

    In this case it is feasible to enumerate all ${7 \choose 4} = 35$ possible arrangements of the disk.
    Doing so we see there are 12 ways to get arrays with one chain of length $2$ and two chains of length $1$.
    We also see that there are 6 ways to get arrays with two chains of length $2$.
    Therefore, $\Pr(C = 2; n = 7, k = 4) = \frac{12}{35}\cdot \frac{1}{3} + \frac{6}{35} \cdot \frac{2}{2}$
\end{example}

For larger values of $n$ and $k$ this quickly becomes infeasible.
Instead, we describe a method whose complexity only depends on $k$.

\begin{lemma}\label{change-lemma}
Let $A$ be an array of size $n$ with $k \leq n$ entries made at random.
Then the probability of a chain of $c$ consecutive changes in $A$ is
\begin{equation}
    \Pr(C = c) = \sum_{p \in P} \frac{{n - k + 1 \choose |p|}}{{n \choose k}}\Pr(C = c \mid p),
\end{equation}
where $P$ is the set of partitions of $k$, $|p|$ is the number of elements in a partition $p$, and $\Pr(C = c \mid p)$ is the probability of $c$ in a partition $p$.
\end{lemma}

\begin{proof}
Let $A$ be an array of length $n$ with $k$ changes made uniformly at random.
Then $A$ can be represented as $p = (p_1, p_2, ..., p_k)$, an ordered partition of $k$,
where each $p_i \in \mathbb{Z}$ represents the $i$-th string of $p_i$ consecutive $1$s separated by one or more $0$s.
Let $|p|$ be the length of the partition of $k$.
By our construction of $p$, $A$ is uniquely represented by $p$ and
\begin{equation}\label{partion-k}
    p_1 + p_2 + ... + p_{|p|} = k.
\end{equation}

Since any array $A$ can be represented by $p \in P$, where $P$ is the partition of $k$, we can compute $\Pr(C = c)$ as
\begin{equation}\label{marginal-probability}
    \Pr(C = c) = \sum_{p \in P} \Pr(C = c \mid p)\Pr(p),
\end{equation}
by marginalizing over $P$, the size of which is ${2k - 1 \choose k - 1}$.
We conclude the proof by computing $\Pr(p)$.
Since there are ${n \choose k}$ possible arrays, counting the number of arrays represented by $p$ is sufficient to compute $\Pr(p)$.

\begin{example}
    (continued from example \ref{a:example}) The ordered partitions of 4 are
    $$(4), (3,1), (1,3), (2,2), (2,1,1), (1,2,1), (1,1,2), (1,1,1,1).$$
    Taking $c=2$, we have $(2,2), (2,1,1), (1,2,1), (1,1,2)$ all contain at least one $2$.

    Notice that when $p = (2,1,1)$, the array will take the form $\star 110 \star 10 \star 1 \star$, where $\star$ represents zero or more $0$s.
    Therefore, there is a single $0$ whose location is not fixed by $p$, and there are ${4 \choose 1} = 4$ ways to place it.

\end{example}

Arrays that are represented by $p$ must have the form
\begin{equation*}
    \star \underbrace{11 \cdots 10}_{p_1} \star \underbrace{11 \cdots 10}_{p_2} \star \cdots \star \underbrace{11 \cdots 1}_{p_{|p|}} \star
\end{equation*}
by our construction.
Notice that there are $k$ $1$s and $|p|-1$ $0$s in the string above,
so there are $n - k - (|p| - 1)$ $0$s whose locations are unfixed.
There are ${n - k -1 \choose |p|}$ different ways to place these $0$s, thus
\begin{equation}
    \Pr(p) = \frac{{n - k - 1 \choose |p|}}{{n \choose k}},
\end{equation}
completing the proof.
\end{proof}

\end{document}